\documentclass[journal,twocolumn]{IEEEtran}
\usepackage[noadjust]{cite}
\usepackage[subsection]{algorithm}
\usepackage{amsmath,amssymb,amsbsy,algorithmic,bm,url,color}
\ifCLASSINFOpdf
   \usepackage[pdftex]{graphicx}
   \DeclareGraphicsExtensions{.pdf,.jpeg,.png}
\else
   \usepackage[dvips]{graphicx}
   \DeclareGraphicsExtensions{.eps}
\fi

\newtheorem{theorem}{Theorem}
\newtheorem{definition}{Definition}

\begin{document}
\title{Diffusion of Cooperative Behavior in Decentralized Cognitive Radio Networks with Selfish Spectrum Sensors}

\author{Amitav~Mukherjee
\thanks{A. Mukherjee is with Hitachi America Ltd., Santa Clara, CA 95050, USA. (e-mail: \tt{ amitav.mukherjee@hal.hitachi.com})}
}
\maketitle

\begin{abstract}
This work investigates the diffusion of cooperative behavior over time in a decentralized cognitive radio network with selfish spectrum-sensing users. The users can individually choose whether or not to participate in cooperative spectrum sensing, in order to maximize their individual payoff defined in terms of the sensing false-alarm rate and transmit energy expenditure. The system is modeled as a partially connected network with a statistical distribution of the degree of the users, who play their myopic best responses to the actions of their neighbors at each iteration. Based on this model, we investigate the existence and characterization of Bayesian Nash Equilibria for the diffusion game. The impacts of network topology, channel fading statistics, sensing protocol, and multiple antennas on the outcome of the diffusion process are analyzed next. Simulation results that demonstrate how conducive different network scenarios are to the diffusion of cooperation are presented for further insight, and we conclude with a discussion on additional refinements and issues worth pursuing.
\end{abstract}

\begin{IEEEkeywords}
Cognitive radio, cooperative spectrum sensing, diffusion, network games, Bayesian Nash Equilibrium.
\end{IEEEkeywords}

\section{INTRODUCTION}\label{sec:intro}

Dynamic spectrum access (DSA) by interweave cognitive radios (ICRs) is emerging as a promising solution to enable better
utilization of the radio spectrum, especially in bands that are currently
under-utilized \cite{Niyato_Han,Jafar09}. DSA partitions wireless terminals into categories of primary (licensed) and secondary (cognitive radio) users, where the primary users (PUs) have priority in accessing the shared spectrum.
ICRs are allowed to opportunistically use the spectrum only when it is not occupied by primary transmitters (PTs) that have priority. Therefore, the ICRs do not cause interference to the PUs in principle.  In the absence of standard control channels or coordinated medium access between the primary and secondary users, the ICRs must periodically sense the spectrum for the presence of PTs and cease transmission upon detection. Local spectrum sensing (LSS) algorithms where each ICR makes an independent decision on whether the spectrum is available have been studied extensively, e.g., \cite{Simon07}--\cite{Zeng10}.

LSS may fail to provide sufficient accuracy due to the vagaries of the wireless medium such as deep fades and shadowing. As a remedy, cooperative spectrum sensing (CSS) has been shown to greatly increase the reliability
of spectrum sensing, with a corresponding increase in complexity and energy consumption \cite{Akyildiz11}. Under CSS,
each ICR either sends its local sensing data/decision to a central collector known as the fusion center \cite{Atapattu11}-\cite{Santucci09}, or shares this information with its neighbors in the case of distributed ICR networks \cite{Huang10}-\cite{Penna12}. However, the fundamental assumption in the CSS studies cited thus far is that all ICRs \emph{willingly engage} in cooperative sensing in order to optimize a global performance metric. In decentralized networks, ICRs can potentially pursue selfish motives and make independent decisions regarding whether to cooperate with their peers via CSS or to act alone by adopting LSS. Therefore, there has been recent interest in game-theoretic models of decentralized networks where ICRs act selfishly in order maximize individual utilities \cite{Cheng11}-\cite{Cheng12}. Yuan \emph{et al.} \cite{Cheng11} study generalized Nash Equilibria of a non-cooperative game where the ICR throughput features in the utility function. Evolutionary game models and associated evolutionary stable strategies with throughput as payoffs are formulated in \cite{Liu10,Letaief09}. Cooperative game-theoretic approach for coalition formation in fixed network topologies are presented in \cite{DaSilva10,Poor12}, while \cite{Cheng12} assumes the presence of a fusion center and ICRs non-cooperatively optimize the frequency over time with which they participate in CSS.  This paper has several major differences from \cite{Cheng11}-\cite{Cheng12} as we analyze Bayesian Nash Equilibria under an imperfect information scenario, base the ICR utility function on the sensing false-alarm rate, and explicitly consider a complex, partially-connected network topology with an arbitrary degree distribution.

On a broader level, there has been extensive recent work on distributed estimation or detection of a single parameter of interest without the need for a central fusion center. Diffusion strategies for distributed estimation/detection in decentralized networks where all users cooperate have been studied in \cite{Sayed08}-\cite{Sayed12}, for example. Consensus algorithms for in-network computation have been presented in \cite{Huang10,Penna12,Barbarossa10} among others. While \cite{Huang10,Penna12,Sayed08}-\cite{Barbarossa10} focus on the \emph{diffusion of information} between collaborating peers, this work is focused on the \emph{diffusion of cooperation} across a decentralized interweave network composed of selfish users. In other words, we are interested in determining conditions under which a given ICR network converges towards or diverges away from global cooperation over time when users individually adapt to the behavior of their peers, with information exchange being implicit.

The motivation to study the diffusion of cooperation is as follows. The autonomous nature of the distributed ICR network makes it difficult for the cognitive network designer to predict the behavior of the ICRs, and whether a steady-state outcome exists for the network as a whole, in terms of what fraction of ICRs end up cooperating. This is a critical question, since the fundamental basis of deploying spectrum-sensing ICRs is to maximize the utilization of unused spectrum - which is achieved when the ICR false-alarm rate is minimized (subject to a primary detection probability constraint).
As the per-ICR false-alarm rate is minimized when \emph{all} ICRs participate in CSS, studying the diffusion of cooperation in the network and the corresponding steady-state properties illustrates how close or far the distributed network is from this ideal outcome. The diffusion analysis in this work shows how key physical and network parameters such as sensing protocol, network degree distribution, number of antennas, shadowing correlation and path loss exponent, etc. impact the steady-state cooperation outcomes of the distributed ICR network.

The paper is organized as follows. Section~\ref{sec:model} introduces the mathematical model of the decentralized ICR network and the spectrum sensing performance of the ICRs. The game-theoretic model and equilibrium properties of the diffusion process are presented in Sec.~\ref{sec:DiffusionGame}. Sec.~\ref{sec:Outcomes} examines the impact of various network and sensing parameters on the diffusion process.
Selected numerical examples are shown in Section~\ref{sec:Sim} followed by a discussion of further research issues in Section~\ref{sec:Discuss}, and we conclude in Section~\ref{sec:Conclusion}.

\emph{Notation}: We will use $\mathcal{N}(\mathbf{m},\mathbf{Z})$ to denote a multivariate Gaussian distribution with mean $\mathbf{m}$ and covariance matrix $\mathbf{Z}$, and define the Gaussian $Q$-function as $Q\left( x \right) = {\left( {\sqrt {2\pi } } \right)^{ - 1}}\int_x^\infty  {{e^{ - {{{u^2}} \mathord{\left/
 {\vphantom {{{u^2}} 2}} \right.
 \kern-\nulldelimiterspace} 2}}}du}$. We also use
$\mathcal{E}\{\cdot\}$ to denote expectation, $(\cdot)^T$ for the transpose, $(\cdot)^H$ for the Hermitian
transpose, $(\cdot)^{-1}$ for the matrix inverse, $\mathbf{x}_{ - i}$ to denote a vector excluding its $i^{th}$ component, ${\left[ {\mathbf{A}} \right]_{i,j}}$ is the $(i,j)$ element of matrix $\mathbf{A}$, and $\mathbf{1}$ is a column vector of all ones.

\section{Mathematical Model}\label{sec:model}
\subsection{Network Model}\label{subsec:Network}
\begin{figure}[htbp]\label{fig:ICRNetwork}
\centering
\includegraphics[width=\linewidth]{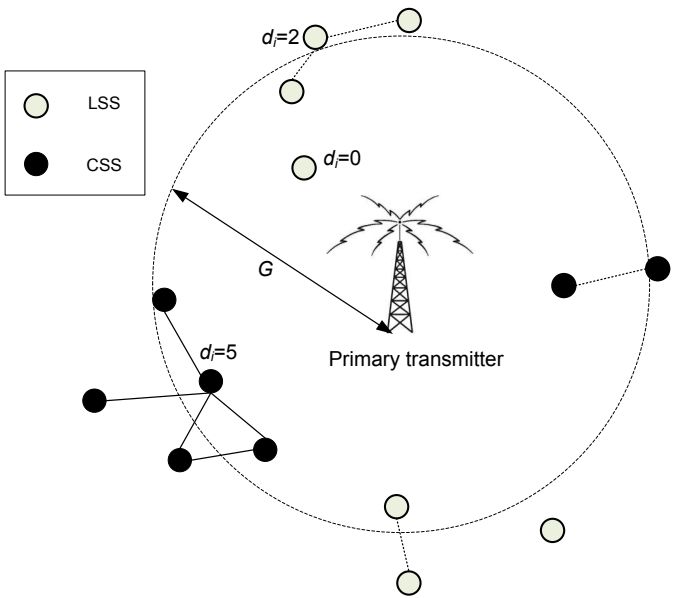}
\caption{Decentralized ICR network with cooperative and non-cooperative users of various degrees. }
\end{figure}

As shown in Fig.~1, the system under consideration is composed of a PT and $n$ single-antenna ICRs that form a partially-connected network represented by the undirected graph $\mathcal{G}\left( {V,E} \right)$, with $V=1,\ldots,n,$ and $E$ being the set of edges or active peer-to-peer links. The $i^{th}$ ICR has a \emph{degree} $d_i$ drawn from a probability distribution $P(d)$ with support $d=0,1,\ldots,D$, such that $\sum\nolimits_{d = 0}^D {P\left( d \right)}  = 1$.  Each ICR is provided prior knowledge of the degree distribution $P(d)$ of the overall network by the network designer and knows its own degree $d_i$, but does not have any knowledge of the degrees of its neighbors. The network is decentralized in the sense that no central fusion center is present, and only localized peer-to-peer interactions are allowed.

At each time instant of the diffusion process (to be defined in Sec.~\ref{sec:DiffusionGame}), the neighbors of each user are a random draw from the population specified by $P(d)$, which indicates that a new network realization is perceived at every interaction epoch\footnote{This assumption is consistent with a mean-field approximation as elaborated upon in Sec.~\ref{sec:time}.}. However, the degree of each ICR remains unchanged over time.
By definition, degree distribution $P(d)$ characterizes the degree of an arbitrary node anywhere in the network. We are now interested in the degree distribution ${\tilde P}(d)$ of a neighboring node, i.e., given some node $i$ with degree $d_i$, what is the degree distribution of its $j^{th}$ neighbor? It is immediately clear that this is not also equal to $P(d)$, since $P(d)$ includes degree-0 nodes in its support that have no neighbors. Instead, consider a randomly chosen edge (link) in the network graph, which by definition exists between neighboring nodes. Once again, the degree of a node we reach by following a randomly chosen edge is not given by $P(d)$.
Since there are $d$ edges that arrive at a node of degree $d$, we are $d$ times as likely
to arrive at that node than another node that has degree 1.
Thus, the probability that the neighboring node is of degree $d$ is proportional to $P(d)d$; normalizing to yield a valid distribution function we obtain \cite{Pintado08}
\begin{equation}
\tilde P\left( d \right) = \frac{{P\left( d \right)d}}{{\sum\nolimits_d {P\left( d \right)d} }}.
\end{equation}

\subsection{Spectrum Sensing Protocol}\label{subsec:SpecSens}
We now specify the spectrum sensing algorithm employed by the ICRs to detect the presence of the PT. Assume that the ICRs are allowed to access the spectrum only if they are located outside a guard region of radius $G$ centered on the PT. Let $y_i$ denote the SNR in dB (logarithmic scale) of the PT signal received at ICR $i$ located at a distance $r_i$ from the PT. Under log-normal shadowing and distance-dependent path loss on the sensing channel, $y_i$ is distributed as a Gaussian random variable with variance $\sigma^2$ and mean $\mu(r_i)$ \cite{Visotsky05,Sousa07}.

For the CSS scenario with a cluster of $c$ cooperating ICRs, the local SNRs are collected into ${{\mathbf{y}} = {{\left[ {{y_1}, \ldots ,{y_c}} \right]}^T}}$ at a randomly chosen participant within the group. We assume each ICR truthfully reports its actual local observation during CSS, and defer discussion of malicious behavior to Sec.~\ref{sec:Discuss}. Since the inter-neighbor distances are much smaller in magnitude compared to $\left\{ {{r_i}} \right\}_{i = 1}^c$, we assume that the SNR reporting is conducted without error or delay, and that the cooperating ICRs are all roughly the same distance from the PT: ${r_1} \approx {r_2} \approx  \ldots  \approx {r_c} = r$ \cite{Visotsky05}-\cite{Veeravalli08}. However, the small intra-cluster size also implies that the SNR observations are correlated, such that the normalized covariance matrix ${\mathbf{\Sigma }} = {\sigma ^{ - 2}}\mathcal{E}\left\{ {{\mathbf{y}}{{\mathbf{y}}^H}} \right\}$ is non-diagonal and of full rank.

Therefore, the binary hypothesis test under CSS to determine whether the cooperating cluster is within the guard region (hypothesis $H_0$) or not (hypothesis $H_1$) is
\begin{align}
  {{H_0}}: &\quad{{\mathbf{y}} \sim \mathcal{N}\left( {\mu \left( r \right){\mathbf{1}},{\sigma ^2}{\mathbf{\Sigma }}} \right)}, \\
  {{H_1}}: &\quad{{\mathbf{y}} \sim \mathcal{N}\left( {\mu \left( {r + \delta } \right){\mathbf{1}},{\sigma ^2}{\mathbf{\Sigma }}} \right)},
\end{align}
for some $\delta > 0$. We assume that both CSS and LSS must achieve a minimum probability of detection target $\beta$, equivalent to a missed-detection probability constraint of $(1-\beta)$. The imposition of the constraint $\beta$ ensures a minimum level of protection for the PUs from unintentional ICR interference. The Neyman-Pearson-optimal likelihood ratio test \cite{KayVolII} that minimizes the probability of false alarm subject to the constraint $\beta$ is given by the decision rule \cite{Visotsky05}
\[\frac{{{{\mathbf{1}}^T}{{\mathbf{\Sigma }}^{ - 1}}{\mathbf{y}}}}{{{{\mathbf{1}}^T}{{\mathbf{\Sigma }}^{ - 1}}{\mathbf{1}}}} \mathop\gtrless \limits_{{{H}_0}}^{{{H}_1}} T\]
with threshold $T = \mu \left( G \right) - {Q^{ - 1}}\left( \beta  \right)$. Define $\Delta \delta  \equiv \mu \left( {R + \delta } \right) - \mu \left( R \right)$, $\Delta \delta <0$. Under an exponential spatial correlation model ${\left[ {\bf{\Sigma }} \right]_{i,j}} = {\rho ^{\left| {i - j} \right|}}$, the corresponding probability of false alarm can be approximated as \cite{Sousa07}
\begin{equation}\label{eq:PFA_CSS}
{P_{FA,CSS}}\left( c \right) = 1 - Q\left( {\frac{{\Delta \delta }}{\sigma }\sqrt {\frac{{\left( {1 - \rho } \right)c + 2\rho }}{{1 + \rho }}}  + {Q^{ - 1}}\left( \beta  \right)} \right).
\end{equation}
The false alarm probability under LSS is easily obtained from \eqref{eq:PFA_CSS} as
\begin{equation}\label{eq:PFA_LSS}
{P_{FA,LSS}} = 1 - Q\left( {\frac{{\Delta \delta }}{\sigma } + {Q^{ - 1}}\left( \beta  \right)} \right).
\end{equation}

At first glance, the primary users appear to be equally protected from ICR interference under LSS and CSS, due to the detection probability constraint $\beta$ imposed on both scenarios. However, from the perspective of the secondary network designer it is always desirable to have as many ICRs cooperate as possible. Since ${P_{FA,CSS}}\left( c \right)$ is monotonically decreasing in $c$, reducing the false-alarm probability via CSS greatly improves the spectrum access opportunities of the ICRs. The reader is referred to \cite{Akyildiz11,Liu11} for specific details on protocols for CSS information exchange and medium access control after the sensing decisions.

\subsection{Cost of Cooperation}
The increased complexity of CSS incurs an additional cost in terms of energy consumption and delay relative to LSS \cite{Cheng11}. Since the diffusion process is iterative by definition, we assume the additional delay due to CSS is negligible compared to the diffusion time scale. The cost of additional energy consumption in reporting SNRs/decisions to a neighbor is modeled as follows. Assume that other ICRs in the proximity of arbitrary user $i$ are its neighbors if they lie within a circle of radius $R$ centered at $i$, with a total of $d_i$ such ICRs. $R$ is thus assumed to be the maximum communication range of the ICRs.
 Let random variable $Y$ represent the distance between neighboring ICRs. The energy consumed by ICR $i$ when it participates in CSS (incurred while either reporting $y_i$ or the CSS decision) is assumed to be proportional to $Y^\alpha$:
  \begin{equation}
  E=cY^\alpha
  \end{equation}
  with $\alpha$ as the path loss exponent of the network, and some proportionality constant $c>0$.

 In sophisticated network geometry models, the user locations are often assumed to be governed by a two-dimensional stochastic point process \cite{Haenggi09}. Assuming $n$ is asymptotically large, we can apply a Poisson point process (PPP) model of example intensity
 \[
 \mu=\pi (R+D)^2/n
 \]
 where $D\gg R$, for which the cdf of the distance to the nearest neighbor \cite{Bettstetter02} is
 \begin{equation}
 {F_Y}\left( y \right) = 1 - {e^{ - \mu \pi {y^2}}}.
 \end{equation}
The cdf of the energy cost is then obtained as
\begin{equation}\label{eq:Energy_cdfPPP}
{F_E}\left( x \right) =1 - {e^{ - \mu \pi {c^{{{ - 2} \mathord{\left/
 {\vphantom {{ - 2} \alpha }} \right.
 \kern-\nulldelimiterspace} \alpha }}}{x^{{2 \mathord{\left/
 {\vphantom {2 \alpha }} \right.
 \kern-\nulldelimiterspace} \alpha }}}}}.
  \end{equation}
We can similarly compute the cost distribution for various other statistical models of the user spatial locations. More generally, in the terminology of Bayesian game theory we can define the degree of a neighboring ICR and energy cost as its `type', which is private information known only to itself. However, the statistical distribution of each user's type is publicly known to all nodes, which is a standard assumption in Bayesian games \cite{Fudenberg}.
Furthermore, note that we can replace the CSS and LSS protocols described above by any arbitrary choice of decision rule (e.g., energy detection, feature detection, hard decision combining) without altering the game-theoretic analysis of the diffusion process in the sequel.

\section{Diffusion Network Game}\label{sec:DiffusionGame}
\subsection{Strategic Non-cooperative Game}
Having delineated the statistical properties of the ICR network and spectrum sensing protocols, we now model the diffusion process as a non-cooperative game that evolves over discrete time steps $t=0,1,\ldots,L$.
Every ICR has a (pure) strategy set of two possible actions $\mathcal{A}=\{0,1\}$ and plays $a\in \mathcal{A}$, where action $a=0$ corresponds to LSS and $a=1$ corresponds to choosing CSS. Let $x^t$ represent \emph{the probability at time $t$ that an arbitrary neighbor anywhere in the network chooses to cooperate} by playing $a=1$. Generally, an ICR may participate in multiple clusters at a time. In that sense $x^t d$ is an approximation of the number of cooperative neighbors. Due to the communication range limit $R$ and the fact that the CSS outcome is broadcast by the local cluster-head, all ICRs within a cluster are assumed to be direct neighbors of each other. The probability that each ICR participates in more than one cluster per time step is therefore considered to be small.
We then define the \emph{utility} obtained by a degree-$d_i$ ICR when it chooses to cooperate as
\begin{equation}\label{eq:utilityCSS}
{u_{{d_i}}}\left( {1,{x^t}} \right) = 1-{P_{FA,CSS}}\left( {{x^t}{d_i}} \right)
\end{equation}
where $P_{FA,CSS}\left( \cdot \right)$ is the cooperative false-alarm rate defined in \eqref{eq:PFA_CSS}. The additional energy cost incurred by ICR $i$ under CSS is represented by $E_i$ drawn from distribution $F_E(x)$ [cf. \eqref{eq:Energy_cdfPPP}], and costs are i.i.d. across ICRs.
Similarly, ICRs abstaining from cooperation obtain
\begin{equation}\label{eq:utilityLSS}
{u_{{d_i}}}\left( {0,{x^t}} \right) = 1-{P_{FA,LSS}}.
\end{equation}
The false-alarm probability is meaningful as a utility function since it has a direct impact on the ICR throughput. Since ICRs must refrain from transmission if they decide that a primary user is active, a higher false alarm rate corresponds to decreased opportunities for spectrum access. Therefore, maximizing the ICR utility function is equivalent to minimizing the false-alarm rate, which is desirable for all ICRs.

The \emph{return function} $v(d_i,x^t)$ represents the additional utility gained by a degree-$d_i$ user that chooses to cooperate:
\begin{equation}\label{eq:revenuefunc}
v\left( {{d_i},{x^t}} \right) = {u_{{d_i}}}\left( {1,{x^t}} \right) - {u_{{d_i}}}\left( {0,{x^t}} \right).
\end{equation}
The user payoff functions are finally given by
\begin{equation}\label{eq:payoffs}
{\Pi _{{d_i}}}\left( {a,{x^t}} \right) = \left\{ {\begin{array}{*{20}{c}}
  {{u_{{d_i}}}\left( {1,{x^t}} \right) - {c_i},}&{a = 1} \\
  {{u_{{d_i}}}\left( {0,{x^t}} \right),}&{a = 0}.
\end{array}} \right.
\end{equation}

From \eqref{eq:utilityCSS}-\eqref{eq:payoffs}, we observe that (i) ${u_{{d_i}}}\left( {0,{x^t}} \right)={u_{{d_i}}}\left( {a,0} \right)$, (ii) ${u_{{d_i}}}\left( {a,{\tilde x^t}} \right)\geq {u_{{d_i}}}\left( {a,x^t} \right)$ if ${\tilde x^t} \geq x^t$ which implies the utilities exhibit positive externalities, and (iii) ${u_{{d_i}}}\left( {1,{x^t}} \right)>{u_{{d_i}}}\left( {0,{x^t}} \right)$ $\forall d_i>0,{x^t}>0$. Therefore, a selfish ICR has an incentive to participate in CSS if and only if
$v\left( {{d_i},{x^t}} \right) \ge E_i$, which occurs with probability
\begin{equation}\label{eq:ProbCoop}
\Pr\{v\left( {{d_i},{x^t}} \right)\} \ge E_i\}=F_E\left(v\left( {{d_i},{x^t}} \right)\right).
\end{equation}

 At each time step, the ICRs play their \emph{myopic best responses}
\begin{equation}\label{eq:BestResponse}
 {b_i}\left( {{{\mathbf{a}}_{ - i}}} \right) = \arg \mathop {\max }\limits_{a \in \mathcal{A}} {\Pi _{{d_i}}}\left( {a,{x^t}} \right) \:\forall i
\end{equation}
 to maximize their individual payoffs, based on the expected behavior of their neighbors parameterized by $x^t$. The diffusion of cooperation through the network is then captured by the evolution of $x^t$ over time. The game model is one of incomplete information since each user is assumed to know only its own degree and cost; the degree and cost realizations of its neighbors are completely unknown. Our game model therefore differs significantly from \cite{Cheng11}, where the utilities are defined in terms of throughput and a joint detection probability constraint is applied to all ICRs under CSS, leading to a coupled strategy space.
The non-cooperative diffusion network game in strategic form is succinctly represented as
${\Gamma ^d}\left( {V,\mathcal{A},{\Pi _{{d_i}}}\left( {a,{x^t}} \right)} \right)$.
For the incomplete information scenario, steady-state rest points of the diffusion process are described by Bayesian Nash Equilibria (BNE) \cite{Fudenberg}, which we characterize next.

\subsection{Structure of Equilibria}
\begin{theorem}\label{Thm:1}
The set of Bayesian Nash Equilibria of the diffusion network game ${\Gamma ^d}\left( {V,\mathcal{A},{\Pi _{{d_i}}}\left( {a,{x^t}} \right)} \right)$ is non-empty.
\end{theorem}
\begin{proof}
The ICR strategy set $\mathcal{A}=\{0,1\}$ is a compact subset over $\mathbb{R}$, or equivalently, a sublattice of $\mathbb{R}$. The return function $v(d_i,x^t)$ is increasing in $x^t$ for each feasible value of $d_i$, which can be verified either by inspection of \eqref{eq:revenuefunc}, or by verifying that its first derivative
\begin{equation}
\begin{split}
  \frac{{dv\left( {{d_i},{x^t}} \right)}}{{d{x^t}}} = & - \frac{1}{{\sqrt {2\pi } }}{e^{\left( { - \frac{{\Delta \delta }}{\sigma }\sqrt {{{\left( {\left( {1 - \rho } \right){x^t}{d_i} + 2\rho } \right)} \mathord{\left/
 {\vphantom {{\left( {\left( {1 - \rho } \right){x^t}{d_i} + 2\rho } \right)} {1 + \rho }}} \right.
 \kern-\nulldelimiterspace} {1 + \rho }}}  + {Q^{ - 1}}\left( \beta  \right)} \right)}}  \\
   &\times \frac{{\Delta \delta }}{{2\sigma \sqrt {1 + \rho } }}\frac{{\left( {1 - \rho } \right){x^t}{d_i}}}{{{{\left( {\left( {1 - \rho } \right){x^t}{d_i} + 2\rho } \right)}^{{1 \mathord{\left/ {\vphantom {1 2}} \right. \kern-\nulldelimiterspace} 2}}}}}
\end{split}
\end{equation}
is positive (recall that $\Delta \delta <0$). The same is true for the individual utility functions and for the payoff ${\Pi _{{d_i}}}\left( {a,{x^t}} \right)$. Thus, the best response correspondence $b_i$ is an increasing function on $\mathcal{A}$ \cite{Yariv07,Vives90}. Define ${\mathbf{b}} = \prod\nolimits_i {{b_i}}$ as the overall best response correspondence obtained from the Cartesian product of the individual responses; ${\mathbf{b}}$ is also an increasing function.

Now, Tarski's fixed point theorem \cite{Tarski55,Fudenberg} states that ``Let $(S;\geq)$ with binary relation $`\geq'$ be a non-empty compact sublattice of $\mathbb{R}^n$ and $f:S\rightarrow S$ an increasing function on $S$, such that for $x,y\in S$, $y\geq x$ implies $f(y)\geq f(x)$. Then the set of fixed points of $f$ is non-empty." Therefore, based on the preceding discussion,
the best response correspondence ${\mathbf{b}}:\prod\nolimits_i\mathcal{A}\rightarrow \prod\nolimits_i\mathcal{A}$ has a non-empty set of fixed points. Since the set of fixed points of a best response correspondence is the set of pure BNE \cite{Fudenberg}, Theorem 1 follows.
\end{proof}

The result in Theorem~\ref{Thm:1} is refined further as follows.
\begin{theorem}\label{Thm:2}
The strategic diffusion network game ${\Gamma ^d}\left( {V,\mathcal{A},{\Pi _{{d_i}}}\left( {a,{x^t}} \right)} \right)$ has a unique BNE in pure strategies.
\end{theorem}
\begin{proof}
For ${\Gamma ^d}\left( {V,\mathcal{A},{\Pi _{{d_i}}}\left( {a,{x^t}} \right)} \right)$ to have a unique BNE, it suffices to show that the probability of cooperation $F_E\left(v\left( {{d_i},{x^t}} \right)\right)$ is concave in $x^t$ for each $d_i$ \cite{Yariv07}. Defining $\eta  = \left( {1 - \rho } \right){x^t}{d_i} + 2\rho$, we first verify that $v\left( {{d_i},{x^t}} \right)$ is concave by computing its second derivative
\[\begin{array}{l}
\frac{{{d^2}v\left( {{d_i},{x^t}} \right)}}{{d{x^t}^2}} =  - \frac{{\Delta \delta {{\left( {\left( {1 - \rho } \right){d_i}} \right)}^2}}}{{2\sigma \sqrt {2\pi } \left( {1 + \rho } \right)}}{e^{ - 0.5{{\left( {\frac{{\Delta \delta }}{\sigma }\sqrt {\frac{\eta }{{\left( {1 + \rho } \right)}}}  + {Q^{ - 1}}\left( \beta  \right)} \right)}^2}}}\\
 \times \left[ {{\eta ^{ - 1}}\left( { - \frac{{\Delta \delta }}{\sigma }\sqrt {\frac{\eta }{{\left( {1 + \rho } \right)}}}  + {Q^{ - 1}}\left( \beta  \right)} \right)\left( {\frac{{\Delta \delta }}{{\sigma \sqrt {\left( {1 + \rho } \right)} }}} \right) - 0.5{\eta ^{ - 1.5}}}\right]
\end{array}\]
which is non-positive for each $d_i$.
For the PPP spatial model, the cost distribution in \eqref{eq:Energy_cdfPPP} is concave in $x$ since the second derivative
\[\frac{{{d^2}{F_E}\left( x \right)}}{{d{x^2}}} = \frac{2}{\alpha }\tilde c{x^{{{\left( {2 - 2\alpha } \right)} \mathord{\left/
 {\vphantom {{\left( {2 - 2\alpha } \right)} \alpha }} \right.
 \kern-\nulldelimiterspace} \alpha }}}{e^{ - \tilde c{x^{{2 \mathord{\left/
 {\vphantom {2 \alpha }} \right.
 \kern-\nulldelimiterspace} \alpha }}}}}\left( {\frac{{2 - \alpha  - 2\tilde c{x^2}}}{\alpha }} \right)\]
is negative, where we have defined ${\tilde c = \mu \pi {c^{{{ - 2} \mathord{\left/
 {\vphantom {{ - 2} \alpha }} \right.
 \kern-\nulldelimiterspace} \alpha }}}}$. $F_E\left(v\left( {{d_i},{x^t}} \right)\right)$ for the PPP model is the composition of two concave functions, therefore it is also concave in $x^t$.
\end{proof}

A fixed-point characterization of the unique BNE can be obtained as follows. Since $\tilde P\left( d \right)$ is the probability of having a neighbor of degree $d$ and ${F_E}\left( {v\left( {{d_i},{x^t}} \right)}\right)$ is the probability that a neighbor cooperates, due to the law of total probability we have
\begin{equation}\label{eq:BNE_fixp}
{x^t} = \phi\left({x^t}\right)\triangleq \sum\nolimits_d {\tilde P\left( d \right){F_E}\left( {v\left( {{d},{x^t}} \right)} \right)}.
\end{equation}
From Theorem~\ref{Thm:2}, the BNE of the diffusion game must satisfy \eqref{eq:BNE_fixp}, and the point satisfying \eqref{eq:BNE_fixp} must be the BNE \cite{Yariv07}.

The final piece of the puzzle relates to the achievability of the unique BNE specified above, which is resolved next. We will make use of the notion of a \emph{supermodular game} with positive externalities, where user payoffs are increasing in
the actions of their neighbors \cite{Goodman02,Johari09}. In other words, in a supermodular game the actions of users mutually reinforce the decisions of their neighbors to follow the same action. A classic example is the power control game in distributed networks; an increase in the perceived interference power at a node triggers it to also increase its own transmit power, and so on \cite{Goodman02}.
\begin{theorem}
The myopic best response dynamics of the selfish ICRs in \eqref{eq:BestResponse} is guaranteed to converge to the unique BNE of ${\Gamma ^d}\left( {V,\mathcal{A},{\Pi _{{d_i}}}\left( {a,{x^t}} \right)} \right)$.
\end{theorem}
\begin{proof}
 Note that ${\Pi _{{d_i}}}\left( {a,{x^t}} \right)$ is upper semi-continuous in $a$ as the player transitions from $a=0$ to $a=1$. Furthermore, the payoff function satisfies the property of having increasing differences in $(a,\mathbf{a}_{-i})$, i.e.,
\begin{equation*}
\begin{split}
{\Pi _{{d_i}}}\left( {\{a'_i,\mathbf{a}'_{-i}\},{x^t}} \right)-{\Pi _{{d_i}}}\left( {\{a_i,\mathbf{a}'_{-i}\},{x^t}} \right)&\geq {\Pi _{{d_i}}}\left( {\{a'_i,\mathbf{a}_{-i}\},{x^t}} \right)\\
&-{\Pi _{{d_i}}}\left( {\{a_i,\mathbf{a}_{-i}\},{x^t}} \right)
\end{split}
\end{equation*}
 for all $a'_i\geq a_i$ and $\mathbf{a}'_{-i}\geq\mathbf{a}_{-i}$. We have already mentioned that $\mathcal{A}$ is a sublattice of $\mathbb{R}$. Therefore, ${\Gamma ^d}\left( {V,\mathcal{A},{\Pi _{{d_i}}}\left( {a,{x^t}} \right)} \right)$ satisfies the properties of a supermodular game with positive externalities \cite{Goodman02,Johari09}, and it is known that myopic best responses always converge to equilibrium in such games \cite{Fudenberg,Honig06}.
\end{proof}

The existence and achievability of a unique BNE as shown above is especially useful since $n$-player supermodular games with two actions per player can have up to $\left\lceil {\frac{n}{2}} \right\rceil$ pure-strategy BNE in general \cite{Dhall09}. The algorithmic description of the Bayesian CSS diffusion game introduced in this section is summarized below.

\begin{algorithm}
\renewcommand{\thealgorithm}{}
{{\footnotesize \caption{Bayesian CSS diffusion game}} \label{alg_mgp}
\begin{algorithmic}
\REQUIRE $0<\epsilon \leq 10^{-3}$, $t=0$
\STATE {\bf Initialization:} \\
\STATE Initialize $x^0$ with a random number between 0 and 1; $t=t+1$ \\
\WHILE{$\left| {\phi \left( {{x^{t + 1}}} \right) - \phi \left( {{x^t}} \right)} \right| > \epsilon $}
\STATE \quad ICRs play best-response strategies \\
\STATE \qquad ${b_i}\left( {{{\mathbf{a}}_{ - i}}} \right) = \arg \mathop {\max }\limits_{a \in \mathcal{A}} {\Pi _{{d_i}}}\left( {a,{x^t}} \right) \:\forall i$\\
\STATE \quad Update $\phi \left( {{x^{t + 1}}} \right)$ via fixed-point eq.~(16); $t=t+1$\\
\ENDWHILE
\end{algorithmic}
}
\end{algorithm}
\subsection{Time Dynamics of Diffusion Process}\label{sec:time}
Having established the properties of steady-state equilibria, it is also of interest to examine the evolution of the diffusion process over time. Since each ICR can be in one of two states (cooperative or non-cooperative) at each time instant, the diffusion process dynamics is a discrete-time Markov chain with $n^2$ possible states. A mean-field approximation alleviates the complexity of such stochastic dynamics by replacing the Markov-chain model with a deterministic discrete-time process. This approximation relies on the stochastic process remaining in the same subset space with a probability
arbitrarily close to one, provided that the population is large enough \cite{Pintado08}. The analytical complexity is further reduced by assuming that at each time instant the neighbors of each ICR are drawn randomly from the population. Therefore, the network degree distribution $P(d)$ can be used to characterize the diffusion process instead of having to account for all possible network connections and topologies, which is nearly intractable \cite{Yariv07,Pintado08,Lamberson10}.

An important metric that captures the evolution of the cooperative behavior of the network is the relative density (fraction) of cooperating ICRs $\xi^t$ at time $t$. By definition, we have
\begin{equation}
{\xi ^t} = \sum\nolimits_k {P\left( k \right)\xi _k^t}
\end{equation}
where $\xi_k^t$ is the relative density of degree-$k$ cooperative ICRs. Assuming the increments in time are arbitrarily small, the dynamic mean-field equation can be written as
\[\frac{{d\xi_k^t}}{{dt}} =  - \xi_k^t\left( {1 - {F_E}\left( {v\left( {k,{x^t}} \right)} \right)} \right) + \xi_k^t{F_E}\left( {v\left( {k,{x^t}} \right)} \right).\]
Setting the derivative to zero yields the stationary condition $\xi_k^t = {F_E}\left( {v\left( {k,{x^t}} \right)} \right)$. Thus at a steady-state BNE, we have
\begin{equation}
{\xi^t_{NE}} = \sum\nolimits_k {P\left( k \right){F_E}\left( {v\left( {k,{x^t}} \right)} \right)},
\end{equation}
which verifies that $\xi^t$ is increasing in $x^t$ due to ${F_E}\left( {v\left( {k,{x^t}} \right)} \right)$ having the same property.

\section{Diffusion Outcomes}\label{sec:Outcomes}
\subsection{Impact of Network Parameters}\label{sec:Impact}
The network parameters that potentially impact the outcome of the diffusion process are
\begin{itemize}
\item Network degree distribution $P(d)$ or $\tilde P(d)$.
\item PT detection constraint $\beta$.
\item Energy cost distribution $F_E(x)$.
\item Sensing protocol.
\item Shadowing correlation and path loss exponent.
\end{itemize}
Changes in any of the above parameters will either lead to a lower or higher equilibrium value of $x^t$. Furthermore, some of these parameters are coupled and cannot be manipulated independently, for e.g., varying the path loss exponent will also alter the energy cost distribution. 
The network mapping $\phi\left({x^t}\right)$ defined in \eqref{eq:BNE_fixp} offers a means of comparing the extent to which cooperation spreads in different networks. Given two networks that differ in one parameter with all others being the same, we now present a framework to evaluate which of the two is more conducive to the diffusion of cooperation.
We will require the following definitions.

\begin{definition}
Given two random variables $A$ and $B$ with distribution functions $F_A(y)$ and $F_B(y)$, if $F_A(y)\geq F_B(y)$ $\forall y\in \mathbb{R}$, then $A$ is first-order stochastically dominated by $B$, denoted as $B \succeq A$.
\end{definition}

\begin{definition}
 A network with map ${\phi'}\left({x^t}\right)$ is said to be \emph{more conducive to diffusion} compared to map $\phi\left({x^t}\right)$ if $\phi\left({x^t}\right)\leq\tilde{\phi}\left({x^t}\right)$ for each $x^t$, since that is equivalent to $\tilde{\phi}\left({x^t}\right)$ having a higher BNE point ${x^t}$. If the converse holds then the network ${\phi'}\left({x^t}\right)$ is less conducive to diffusion.
\end{definition}

We can then summarize the impact of varying a specific network parameter while fixing the remainder as follows.
\begin{itemize}
\item {\bf{Varying $\tilde P(d)$}}: If $\tilde P(d) \succeq \tilde P'(d)$, then $\phi\left({x^t}\right)=\sum\nolimits_d {\tilde P\left( d \right){F_E}\left( {v\left( {{d},{x^t}} \right)} \right)}\geq  \sum\nolimits_d {\tilde P'\left( d \right){F_E}\left( {v\left( {{d_i},{x^t}} \right)} \right)}=\tilde{\phi}'\left({x^t}\right)$. Increasing the probability of higher-degree neighbors aids diffusion.
\item \textbf{Varying $\beta$}: The false-alarm probability of any feasible decision rule is non-decreasing in $\beta$ (property of any receiver operating characteristic). For $\beta'\leq \beta$, we have $P'_{FA,CSS}(c)\leq P_{FA,CSS}(c)$ and $P'_{FA,LSS}\leq P_{FA,LSS}$. This implies $v'(d_i,x^t)\geq v(d_i,x^t)$ and ${F_E}\left( {v'\left( {{d_i},{x^t}} \right)}\right)\geq {F_E}\left( {v\left( {{d_i},{x^t}} \right)}\right)$, therefore ${\phi}'\left({x^t}\right)\geq {\phi}\left({x^t}\right)$.
    Lowering the detection probability constraint increases the gain from CSS while the cost remains unchanged, thus enhancing diffusion.
\item {\textbf{Varying $F_E(x)$}}: It is obvious that for $F_E(x)\succeq F'_E(x)$, we have ${\phi}'\left({x^t}\right)\leq {\phi}\left({x^t}\right)$. Increasing the cost of cooperation can only hinder the diffusion process.
\item {\textbf{Varying the sensing protocol}}: Since the sensing protocols in Sec.~\ref{subsec:SpecSens} are Neyman-Pearson optimal, for any competing sensing protocol with metrics $P'_{FA,CSS}(c)$ and $P'_{FA,LSS}$, we must have $P'_{FA,CSS}(c)\geq P_{FA,CSS}(c)$ and $P'_{FA,LSS}\geq P_{FA,LSS}$ for the same detection constraint $\beta$. This implies $v'(d_i,x^t)\leq v(d_i,x^t)$ and ultimately ${\phi}'\left({x^t}\right)\leq {\phi}\left({x^t}\right)$. Deviating from the optimal likelihood ratio test decreases the likelihood of cooperation.
\item \textbf{Varying $\rho$}: For shadowing correlation coefficients $\rho$ and $\rho'$, if $\rho'\geq\rho$, since the term $\sqrt {{{\left( {\left( {1 - \rho } \right)c + 2\rho } \right)} \mathord{\left/
 {\vphantom {{\left( {\left( {1 - \rho } \right)c + 2\rho } \right)} {1 + \rho }}} \right.
 \kern-\nulldelimiterspace} {1 + \rho }}}$ in \eqref{eq:PFA_CSS} is decreasing in $\rho$ we obtain ${\phi}'\left({x^t}\right)\leq {\phi}\left({x^t}\right)$. Increased spatial correlation diminishes the gain from CSS and makes cooperation less likely.
\item \textbf{Varying $\alpha$}: For path loss exponents $\alpha,\alpha'$, if $\alpha\geq\alpha'$ then $F_E(x)\succeq F'_E(x)$ and $\tilde{\phi}'\left({x^t}\right)\geq \tilde{\phi}\left({x^t}\right)$. Increasing the cost distribution clearly inhibits diffusion.
\end{itemize}

\subsection{Impact of Multiple Antennas}
While the development thus far has considered the case of single-antenna users, it is worthwhile to investigate the multiple-input multiple-output (MIMO) scenario where the ICRs are equipped with $M$ antennas each. The most notable impact is on the performance of the LSS and CSS likelihood ratio tests.
Specifically, we must now account for the correlation between antennas at each ICR, in addition to the spatial correlation across the ICRs.

Let the $(M \times 1)$ received vector at the $i^{th}$ ICR be
\[\mathbf{z}_i=\mathbf{s}_i+\mathbf{n}_i\]
 where $\mathbf{s}_i$ is the observation of the primary signal and $\mathbf{n}_i\sim \mathcal{CN}(\mathbf{0},\sigma_n^2\mathbf{I})$ is complex additive Gaussian noise independent of $\mathbf{s}_i$. Due to the effect of antenna correlation, ${{\mathbf{\Sigma }}_{s,i}} = E\left\{ {{{\mathbf{s}}_i}{\mathbf{s}}_i^H} \right\}$ is a non-diagonal Hermitian matrix. Since each ICR reports its local SNR under CSS, the SNR-maximizing strategy in the multi-antenna scenario is to coherently combine the SNRs measured at each of the $M$ antennas via maximum-ratio combining (MRC).

 Given the eigenvalue decomposition ${{\mathbf{\Sigma }}_{s,i}} = {\mathbf{U\Lambda }}{{\mathbf{U}}^H}$, the optimal combining rule is to first decorrelate the received signal as ${{{\mathbf{\tilde z}}}_i} = {{\mathbf{U}}^H}{{\mathbf{z}}_i}$, followed by MRC to yield the output SNR ${\tilde\gamma _i} = {{{\mathbf{\tilde z}}_i^H{{{\mathbf{\tilde z}}}_i}} \mathord{\left/
 {\vphantom {{{\mathbf{\tilde z}}_i^H{{{\mathbf{\tilde z}}}_i}} {2\sigma _n^2}}} \right.
 \kern-\nulldelimiterspace} {2\sigma _n^2}}$ \cite{Beaulieu02}. Since $\mathbf{U}$ is unitary, the distribution of ${\tilde\gamma _i}$ coincides with that of the combiner output SNR ${\gamma _i} = {{{\mathbf{z}}_i^H{{\mathbf{z}}_i}} \mathord{\left/
 {\vphantom {{{\mathbf{z}}_i^H{{\mathbf{z}}_i}} {2\sigma _n^2}}} \right.
 \kern-\nulldelimiterspace} {2\sigma _n^2}}$ in the uncorrelated antenna scenario. The MRC output SNR ${\gamma _i}$ is the sum of the $M$ i.i.d. per-antenna SNRs, each of which is a Gaussian random variable (in dB) with mean $\mu(r)$ and variance $\sigma^2$, which implies ${\gamma _i}\sim \mathcal{N}(M\mu(r),M\sigma^2)$. Assuming error-free SNR reports as before, the binary hypothesis test under MIMO CSS is
\begin{align*}
  {{H_0}}:&\quad{{\mathbf{y}} \sim \mathcal{N}\left( M{\mu \left( r \right){\mathbf{1}},{\sigma ^2}M{\mathbf{\Sigma }}} \right)},   \\
  {{H_1}}:&\quad{{\mathbf{y}} \sim \mathcal{N}\left( M{\mu \left( {r + \delta } \right){\mathbf{1}},{\sigma ^2}M{\mathbf{\Sigma }}} \right)}.
\end{align*}
The corresponding probability of false alarm can then be written as
\begin{equation}\label{eq:PFA_CSSMIMO}
{P^{M}_{FA,CSS}}\left( c \right) = 1 - Q\left( {\frac{\sqrt{M}{\Delta \delta }}{\sigma }\sqrt {\frac{{\left( {1 - \rho } \right)c + 2\rho }}{{1 + \rho }}}  + {Q^{ - 1}}\left( \beta  \right)} \right)
\end{equation}
and the false alarm probability under MIMO LSS is
\begin{equation}\label{eq:PFA_LSSMIMO}
{P^M_{FA,LSS}} = 1 - Q\left( {\frac{\sqrt{M}{\Delta \delta }}{\sigma } + {Q^{ - 1}}\left( \beta  \right)} \right).
\end{equation}
Assume that the SNR/decision reports are conducted using one out of the $M$ antennas, such that the cost distribution is unchanged. It then follows that $v^M(d_i,x^t)\geq v(d_i,x^t)$ and ${\phi}^M\left({x^t}\right)\geq {\phi}\left({x^t}\right)$, where the superscript `M' denotes metrics of the MIMO scenario. In other words, deploying multiple antennas at the ICRs \emph{enhances} the diffusion of cooperation in the network. Furthermore, the trends and conclusions derived in Sec.~\ref{sec:Impact} can be shown to also hold for the MIMO ICR network.

\section{Simulation Results}\label{sec:Sim}
In this section, we present the results of several numerical
experiments that investigate how conducive different decentralized ICR networks are to the diffusion of cooperation. 
The background AWGN
variance at all receivers is assumed to be unity. In addition, unless specified otherwise we set the number of ICRs to $n=18$, number of antennas per ICR as $M=1$, the network degree distribution as $[\begin{array}{*{20}{c}}
  {P\left( 1 \right) = 0.37}&{P\left( 2 \right) = 0.33}&{P\left( 3 \right) = 0.3}
\end{array}]$, $D=20$m, the detection probability target to $\beta=0.95$, path-loss exponent $\alpha=2.5$, proportionality constant $c=2$, shadowing correlation $\rho={e^{{{ - 0.1R} \mathord{\left/ {\vphantom {{ - 0.1R} {\left( {n - 1} \right)}}} \right.
 \kern-\nulldelimiterspace} {\left( {n - 1} \right)}}}}$ \cite{Sousa07}, sensing parameters $\delta=-0.09,\sigma=3.3$ \cite{Sousa07}, range $R=2$m, and initial cooperation probability $x^0=0.3$.

\begin{figure}[htbp]
\centering
\includegraphics[width=\linewidth]{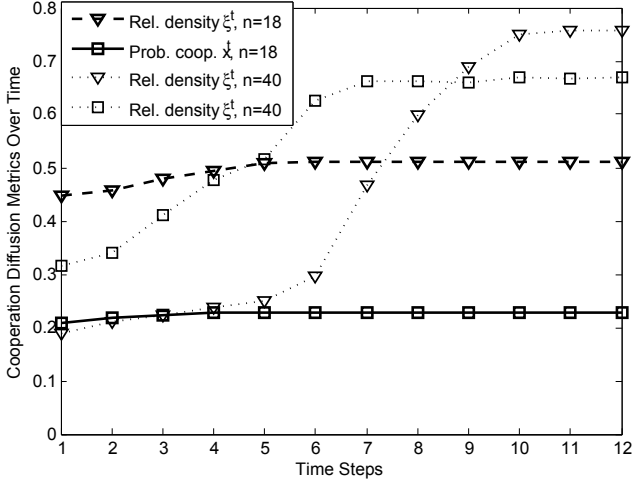}
\caption{Evolution of probability of cooperation and relative density of cooperating ICRs over time $t$.}
\label{fig_vst}
\end{figure}
Fig.~\ref{fig_vst} presents the evolution of diffusion metrics $x^t$ and $\xi^t$ over time for $n=18,R=2$m, and $n=40,R=1.25$m, assuming $M=1$ and an initial value of $x^0=0.2$. We observe that for the case $n=18$, $x^t$ and $\xi^t$ rapidly attain their equilibrium values within four and six time steps, respectively, and do not exhibit major changes from the initial values. On the other hand, the scenario of $n=40$ exhibits much more dramatic increases over time since CSS is more attractive given the lower energy cost due to $R$ being smaller. The cooperation metrics increase monotonically over time, despite ICRs having the freedom to switch back to LSS from CSS in previous time instants. This implies that cooperation is a mutually reinforcing behavior for ICRs for whom the gain from CSS outweighs the cost, in spite of their myopic view of the network.

\begin{figure}[htbp]
\centering
\includegraphics[width=\linewidth]{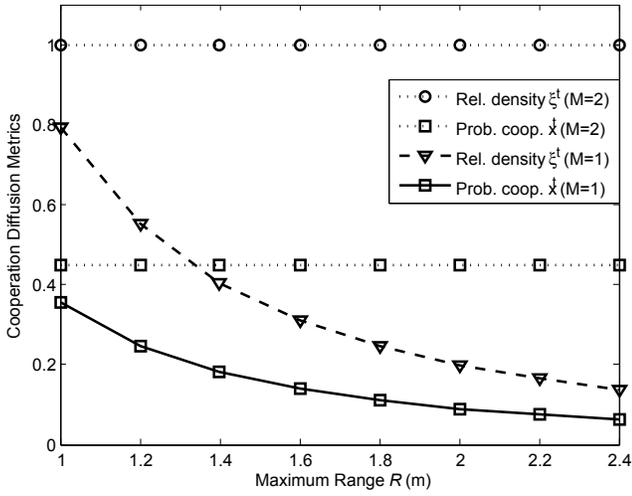}
\caption{Equilibrium probability of cooperation and relative density of cooperating ICRs versus $R$.}
\label{fig_vsR}
\end{figure}
In Fig.~\ref{fig_vsR}, equilibrium metrics $x^t$ and $\xi^t$ are displayed as the maximum inter-ICR communication range $R$ increases. As discussed in Sec.~\ref{sec:Outcomes}, an increase $R$ raises the cost of cooperation while the benefit from CSS remains unchanged, and this is evident from the sharp decrease from 80\% to 18\% in relative density for $M=1$. Interestingly, when ICRs are equipped with an additional antenna, the increase in $R$ is not enough to offset the gain from CSS, and in fact global cooperation is the equilibrium result for the scenario of $M=2$.

\begin{figure}[htbp]
\centering
\includegraphics[width=\linewidth]{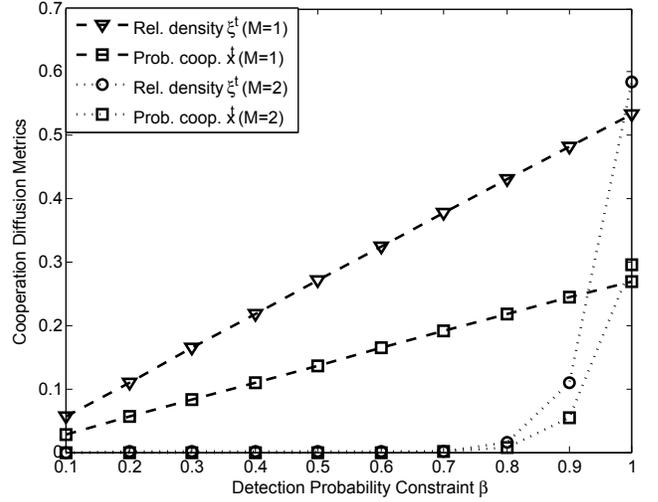}
\caption{Equilibrium probability of cooperation and relative density of cooperating ICRs versus $\beta$.}
\label{fig_vsbeta}
\end{figure}
In Fig.~\ref{fig_vsbeta}, equilibrium metrics $x^t$ and $\xi^t$ are displayed for $M=1,2,$ as a function of the primary user detection probability constraint $\beta$. Interestingly, for less stringent values of $\beta$, MIMO ICRs tend to abstain from cooperation since their local false-alarm rate is correspondingly low and the benefit of cooperation does not outweigh the cost. As $\beta$ becomes more stringent, there is clearly a change in MIMO ICR behavior for $\beta\geq0.8$ where they become much more conducive to CSS.

\begin{figure}[htbp]
\centering
\includegraphics[width=\linewidth]{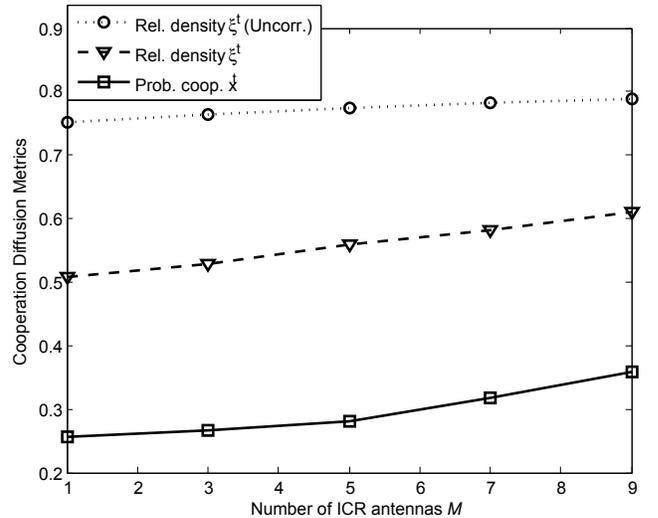}
\caption{Equilibrium probability of cooperation and relative density of cooperating ICRs versus $M$.}
\label{fig_vsM}
\end{figure}
In Fig.~\ref{fig_vsM}, equilibrium metrics $x^t$ and $\xi^t$ are displayed as the number of ICR antennas is increased. The shadowing correlation assumption strongly inhibits the increase of diffusion as $M$ goes from one to three, in contrast to the ideal uncorrelated case of $\rho=0$. Since $M\leq 3$ for practical mobile terminals, only minor gains in cooperative behavior are apparent from deploying multiple antennas in a realistic fading environment.

\section{Discussion}\label{sec:Discuss}
In this section we touch upon several directions for further research, ranging from network topology optimization to more sophisticated ICR communication models and best-response strategies.
\subsection{Network Optimization}
From the perspective of the ICR network designer, it may be desirable to achieve a target equilibrium relative density $\xi^t$ that corresponds to a satisfactory ICR spectrum usage efficiency (in terms of network-wide false alarm rate). Of the parameters listed in Sec.~\ref{sec:Impact}, the majority are either determined by nature ($\rho$,$\alpha$) or by spectrum regulations ($\beta$, sensing protocol), which leaves the network degree distribution $P(d)$ as a metric that could be manipulated after ICR deployment. There are potentially an infinite number of feasible degree distributions that yield a target equilibrium density $\xi^t$.
Bettstetter \cite{Bettstetter02} states that for a 2D-PPP spatial model with intensity $\alpha$, the probability that an arbitrary node has degree $k$ for given range $R$ is
\[\Pr \left( {{d_i} = k} \right) = \frac{{{{\left( {\alpha \pi {R^2}} \right)}^k}}}{{k!}}{e^{ - \alpha \pi {R^2}}}\]
from which we can determine the required range and the corresponding ICR transmit power (assumed to be proportional to $R^{-\alpha}$) such that this probability is arbitrarily close to 1.
Therefore, a n\"{a}ive approach for creating a desired degree distribution $P(d)$ is to assign corresponding transmit powers in the same proportion. Needless to say, such an approach neglects intra-ICR interference, and a general solution remains an open problem (see \cite{Swami11,Ramanathan00} for more sophisticated analyses of topology control in ad hoc networks).

\subsection{Imperfect Inter-ICR Reports}
While shadowing and noise impairments are considered on the sensing channels, the assumption of noiseless inter-neighbor reporting channels in this paper provides an upper bound on the diffusion level of a decentralized ICR network. This is because imperfect SNR reports will degrade the performance of CSS, which in turn diminishes the return function in \eqref{eq:revenuefunc} and thus the probability of cooperation in \eqref{eq:ProbCoop}. Nonetheless, low-rate SNR reports with sufficient error control coding can approach the performance promised by the ideal reporting channel assumption.

A related question worth answering is the following: given the selfish assumption for the individual ICRs, do they have an incentive to behave maliciously by intentionally reporting false SNR values to their neighbors? The impact and detection of falsified reports in cooperative spectrum sensing have been studied extensively in recent literature \cite{Cabric10,Cabric12}. In our system model, if a malicious ICR overhears its neighbor's reports and assumes they are truthful, it can combine them with its own observation to privately compute $P_{FA,CSS}$ while publicly reporting a false value. If these false reports increase the publicly-known probability of false alarm, then the gain from CSS is diminished and n\"{a}ive ICRs are less likely to choose cooperation. In the next iteration, the malicious ICR is then less likely to receive SNR reports from its neighbors, which discourages future false SNR reports and shows that the diffusion game intrinsically rewards truthful behavior. The security aspects of diffusion-based CSS therefore invite further study, especially since the reputation-based and statistical false report detection techniques in \cite{Cabric10,Cabric12} may not be applicable under the mean-field assumption.

\subsection{Beyond Myopic Strategies}
We have seen that the myopic best response strategy in \eqref{eq:BestResponse} converges to the BNE of the diffusion game. On the other hand, the adage ``the more you know the better you can do" certainly holds true for the diffusion game as well. Specifically, if ICRs can obtain more information regarding their local neighborhood, they can conceivably make better choices regarding whether to take part in CSS. We briefly mention two such possibilities. Firstly,
side information regarding the degrees of an ICR's neighbors can be exploited by choosing to perform CSS only with high-degree neighbors, since higher-degree nodes enhance diffusion as shown in Sec~\ref{sec:Impact}. The best response strategy would then become dependent on the degrees of the neighbors \cite{Johari09}. Secondly, additional side information regarding the cost realized at its neighbors can also be useful to an ICR. Since the cost is assumed to be realized once at the initialization of the diffusion process, an ICR can exploit this information to predict future values of $x^t$ and construct a best-response strategy with memory, as opposed to the memoryless myopic best response approach in the current model.

\section{Conclusion}\label{sec:Conclusion}
This work investigates the diffusion of cooperative behavior over time in a decentralized cognitive radio network with selfish spectrum-sensing users. The system is modeled as a partially connected network with a statistical distribution of the degree of the users, who play their myopic best responses to the actions of their neighbors at each iteration. We proved the existence of an unique Bayesian Nash Equilibrium for the diffusion game, provide its fixed-point characterization, and show it is achievable with myopic best responses. The impacts of network topology, channel fading statistics, and sensing protocol on the outcome of the diffusion process have been examined. Simulation results demonstrate how conducive different network scenarios are to the diffusion of cooperation, and we conclude with a discussion on additional refinements to reporting channel models, security aspects, and best response strategies.


\end{document}